\documentclass[conference]{IEEEtran}
\makeatletter
\def\ps@headings{%
\def\@oddhead{\mbox{}\scriptsize\rightmark \hfil \thepage}%
\def\@evenhead{\scriptsize\thepage \hfil \leftmark\mbox{}}%
\def\@oddfoot{}%
\def\@evenfoot{}}
\makeatother
\usepackage{graphicx}
\usepackage{amssymb}
\usepackage{epstopdf}
\usepackage{amsmath}
\usepackage{enumerate}
\usepackage{cite}
\usepackage{verbatim}
\usepackage{mathcomp}
\usepackage{supertabular}
\usepackage{longtable}
\usepackage{stmaryrd}
\usepackage{url}
\usepackage{mathrsfs}
\usepackage[singlelinecheck=false, font = small, justification = raggedright]{caption}[2008/04/01]
\newcommand{\E}{{\mathcal E}}

\newcommand{\G}{{\mathcal G}}
\newcommand{\HH}{{\mathcal H}}
\newcommand{\I}{{\mathcal I}}
\newcommand{\J}{{\mathcal J}}
\newcommand{\LL}{{\mathcal L}}

\newcommand{\X}{{\mathcal X}}
\newcommand{\Y}{{\mathcal Y}}
\newcommand{\V}{{\mathcal V}}

\newcommand{\N}{{\mathcal N}}

\newcommand{\cC}{{\mathscr{C}}}

\newcommand{\bM}{{\boldsymbol M}}

\newcommand{\bHi}{{\boldsymbol H}^{(i)}} 

\newcommand{\bL}{{\boldsymbol L}}

\newcommand{\supp}{{\sf supp}}\newcommand{\dist}{{\mathsf{d}}}
\newcommand{\weight}{{\mathsf{wt}}}
\newcommand{\spn}{{\mathsf{span}}}
\newcommand{\clspn}{{\mathsf{colspan}}}

\newcommand{\rank}{{\mathsf{rank}}}
\newcommand{\mr}{{\text{min-rank}}}
\newcommand{\mrt}{{\text{min-rank}_2}}
\newcommand{\mrq}{{\text{min-rank}_q}}

\newcommand{\define}{\stackrel{\mbox{\tiny $\triangle$}}{=}}

\newcommand{\bu}{{\boldsymbol u}}
\newcommand{\bv}{{\boldsymbol v}}

\newcommand{\by}{{\boldsymbol y}}

\newcommand{\bc}{{\boldsymbol c}}

\newcommand{\bO}{{\boldsymbol 0}}

\newcommand{\bx}{{\boldsymbol{x}}}
\newcommand{\bz}{{\boldsymbol{z}}}

\newcommand{\bG}{{\boldsymbol{G}}}
\newcommand{\be}{{\boldsymbol e}}

\newcommand{\bep}{{\boldsymbol \epsilon}}
\newcommand{\beph}{\hat{\boldsymbol \epsilon}}

\newcommand{\bbti}{{{\boldsymbol \beta}_i}}

\newcommand{\bvi}{{\boldsymbol v}_{i}}
\newcommand{\Na}{{N_q[\alpha(\HH),2\delta+1]}}
\newcommand{\Nk}{{N_q[\kappa_q(\HH),2\delta+1]}}
\newcommand{\NX}{\N_q(\HH,\delta)}
\newcommand{\aG}{{\alpha(\HH)}}
\newcommand{\kG}{{\kappa_q(\HH)}}
\newcommand{\IX}{{\I(q,\HH)}}
\newcommand{\JX}{{\J(\HH)}}

\newcommand{\bxh}{{\hat{\boldsymbol{x}}}}

\newcommand{\fkD}{{\mathfrak D}}
\newcommand{\fkE}{{\mathfrak E}}

\newcommand{\al}{\alpha}

\newcommand{\kp}{\kappa}

\newcommand{\seq}{\subseteq}
\newcommand{\mic}{{$\HH$-IC}}
\newcommand{\dd}{{{$(\delta,\HH)$-ECIC}}}

\newcommand{\ra}{\rightarrow}
\renewcommand{\ge}{\geqslant}
\renewcommand{\le}{\leqslant}

\newcommand{\ff}{\mathbb{F}}
\newcommand{\fq}{\mathbb{F}_q}
\newcommand{\fqn}{\mathbb{F}_q^n}

\newtheorem{corollary}{Corollary}[section]
\newtheorem{definition}{Definition}[section]
\newtheorem{example}{Example}[section]
\newtheorem{remark}{Remark}[section]
\newtheorem{theorem}{Theorem}[section]
\newtheorem{proposition}[theorem]{Proposition}
\newtheorem{lemma}[theorem]{Lemma}

\title{Index Coding and Error Correction}

\date{}                                           

\begin{document}

\author{
  \IEEEauthorblockN{Son Hoang Dau, Vitaly Skachek, and Yeow Meng Chee}
  \IEEEauthorblockA{Division of Mathematical Sciences,
    School of Physical and Mathematical Sciences\\
    Nanyang Technological University,
    21 Nanyang Link, Singapore 637371\\
    Emails: {\tt \{ DauS0002, Vitaly.Skachek, YMChee \} @ntu.edu.sg}}
}

\maketitle

\begin{abstract}
\boldmath
A problem of index coding with side information was first considered by Y. Birk and T. Kol \emph{(IEEE INFOCOM, 1998)}.
In the present work, a generalization of index coding scheme, where transmitted symbols are subject to errors, is studied. 
Error-correcting methods for such a scheme, and their parameters, are investigated. 
In particular, the following question is discussed:
given the side information hypergraph of index coding scheme and the maximal number of erroneous symbols $\delta$, 
what is the shortest length of a linear index code, such that every receiver is able to recover the required information? 
This question turns out to be a generalization of the problem of finding a shortest-length 
error-correcting code with a prescribed error-correcting capability 
in the classical coding theory. 

The Singleton bound and two other bounds, referred to as the $\al$-bound and
the $\kp$-bound, for the optimal length of a linear error-correcting index code (ECIC)
are established. For large alphabets, a construction based on concatenation of an optimal index 
code with an MDS classical code, is shown to attain the Singleton bound. 
For smaller alphabets, however, this construction may not be optimal. 
A random construction is also analyzed. It yields another inexplicit 
bound on the length of an optimal linear ECIC. 
Finally, the decoding of linear ECIC's is discussed. The syndrome decoding is shown to output the exact 
message if the weight of the error vector is less or equal to the error-correcting capability of the 
corresponding ECIC.  
\end{abstract}

\section{Introduction}
\label{sec:introduction}

\subsection{Background}

\PARstart{T}he problem of Index Coding with Side Information (ICSI) was introduced by Birk and Kol \cite{BirkKol98}.
During the transmission, each client might miss a certain part of the data, due to intermittent reception, limited storage capacity or any other reasons. Via a slow backward channel, the clients let the server know which messages they already have in their possession, and which messages they are interested to receive. The server has to find a way to deliver to each client all the messages he requested, yet spending a minimum number of transmissions. As it was shown in \cite{BirkKol98}, the server can significantly reduce the number of transmissions by coding the messages. 

Possible applications of index coding include communications scenarios, in which a satellite or a server broadcasts a set of messages to a set clients, such as daily newspaper delivery or video-on-demand. Index coding with side information can also be used in opportunistic wireless networks~\cite{Katti2006}. 

The ICSI problem has been a subject of several recent studies \cite{Yossef, Yossef-journal, LubetzkyStav, Rouayheb2007, Rouayheb2008, Alon}. This problem can be viewed as a special case of the Network Coding (NC) problem~\cite{Ahlswede},~\cite{KoetterMedard2003}. In particular, as it was shown in~\cite{Rouayheb2008}, every instance of the NC problem can be reduced to an instance of the ICSI problem.

\subsection{Our contribution}

In this work, we generalize the ICSI problem towards a setup with error correction. 
We extend some known results on index coding to a case where any receiver 
can correct up to a certain number of errors. The problem of designing such error-correcting 
index codes (ECIC's) naturally generalizes the problem of constructing classical error-correcting codes. 
We establish an upper bound (the $\kp$-bound) and a lower bound (the $\al$-bound) 
on the shortest length of a linear ECIC, which is able to correct any error pattern of size up to $\delta$. 
We also derive an analog of the Singleton bound, and show that this bound is tight for codes over large alphabets. 
We also consider random ECIC's. By analyzing their parameters, we obtain an upper bound on their length. 
Finally, we discuss the decoding of linear ECIC's.
We show that the syndrome decoding results in a correct result, provided that the number of errors does not exceed
the error-correcting capability of the code.       

The problem of error correction for NC was studied in several previous works. However, these results are not 
directly applicable for the ICSI problem. First, the existing works only consider the multicast scenario, 
while the ICSI problem, however, is a special case of the non-multicast 
NC problem. Second, the ICSI problem can be modeled by the NC scenario~\cite{Alon}, 
yet, this requires that there are directed edges from particular sources to each sink, 
which provide the side information. The symbols transmitted on these special edges, unlike for error-correcting NC, are not allowed to be corrupted.

\section{Preliminaries}
\label{sec:preliminaries}

Let $\fq$ be the finite field of $q$ elements, where $q$ is a power of prime, and $\fq^* = \fq \backslash \{0\}$. 
Let $[n] = \{1,2,\ldots,n\}$. 
For the vectors $\bu, \bv \in \fq^n$, we use $\dist(\bu,\bv)$ to denote the
the Hamming distance between $\bu$ and $\bv$. 
If $\bu \in \fq^n$ and $\bM \subseteq \fq^n$ is a set of vectors (or a vector subspace), 
then this notation can be extended to 
\[
\dist(\bu,\bM) = \min_{\bv \in \bM} \dist(\bu, \bv) \; . 
\]
Given $q$, $k$, and $d$, let $N_q[k,d]$ denote the length of the shortest linear code over $\fq$ which has dimension $k$ and minimum distance $d$.  
The \emph{support} of a vector $\bu \in \fqn$ is defined by $\supp(\bu) \define \{i \in [n]: u_i \neq 0\}$.
The Hamming weight of $\bu$ is defined by $\weight(\bu) \define |\supp(\bu)|$.  
Suppose $E \subseteq [n]$. We write $\bu \lhd E$ whenever $\supp(\bu) \seq E$.  

We use $\be_i = (\underbrace{0,\ldots,0}_{i-1},1,\underbrace{0,\ldots,0}_{n-i}) \in \fqn$ to denote the unit vector, which has a one at the $i$th position, and zeros elsewhere.
For a vector $\by = (y_1,y_2,\ldots,y_n)$ and a subset $B = \{i_1,i_2,\ldots,i_b\}$ of $[n]$, where $i_1 < i_2 < \cdots <i_b$, let $\by_B$ denote the vector $(y_{i_1},y_{i_2},\ldots,y_{i_b})$. 

For an $n \times N$ matrix $\bL$, let $\bL_i$ denote its $i$th row. For a set $E \subseteq [n]$, 
let $\bL_E$ denote the $|E| \times N$ matrix obtained from $\bL$ by deleting all the rows of $\bL$ which are not indexed by the elements of $E$. 
For a set of vectors $\bM$, we use notation $\spn(\bM)$ to denote the linear space spanned by the vectors in $\bM$.  
We also use notation $\clspn(\bL)$ for the linear space spanned by the columns of the matrix $\bL$. 

Let $\G = (\V, \E)$ be a graph with a vertex set $\V$ and an edge set $\E$.
A directed graph $\G$ is called \emph{symmetric} if 
\[
(u, v) \in \E \quad \Leftrightarrow \quad (v, u) \in \E \; . 
\] 
The \emph{independence number} of an undirected graph $\G$ is denoted by $\al(\G)$. 
There is a natural correspondence between undirected graphs and 
directed symmetric graphs.
By using this correspondence, the definition of independence number
is naturally extended to directed symmetric graphs.

\section{Error-Correcting Index Coding with Side Information}
\label{sec:ic_ecc}

Index Coding with Side Information problem considers the following communications scenario. 
There is a unique sender (or source) $S$, who has a vector of messages 
$\bx = (x_1, x_2, \ldots, x_n)$ in his possession. 
There are also $m$ receivers $R_1,R_2,\ldots,R_m$, receiving information from $S$ via a broadcast channel. 
For each $i \in [m]$, $R_i$ has side information, 
i.e. $R_i$ owns a subset of messages $\{ x_j \}_{j \in \X_i}$, where $\X_i \subseteq [n]$. 
Each $R_i$, $i \in [m]$, 
is interested in receiving the message $x_{f(i)}$ (we say that $R_i$ requires $x_{f(i)}$), 
where the mapping $f: [m] \ra [n]$ satisfies $f(i) \notin \X_i$ for all $i \in [m]$. 
Hereafter, we use the notation $\X = (\X_1, \X_2, \ldots, \X_m)$. 
An instance of the ICSI problem is given by a quadruple $(m,n,\X,f)$. 
An instance of the ICSI problem can also be conveniently described by the following directed hypergraph~\cite{Alon}. 

\vskip 10pt
\begin{definition}
Let $(m, n, \X, f)$ be an instance of the ICSI problem.  
The corresponding \emph{side information (directed) hypergraph} $\HH = \HH(m,n,\X,f)$ is defined by the vertex set
$\V = [n]$ and the edge set $\E_\HH$, where 
\[
\E_\HH = \{ (f(i), \X_i) \; : \; i \in [n]\} \; . 
\]
We often refer to $(m,n,\X,f)$ as an instance of the ICSI problem described by the hypergraph $\HH$. 
\end{definition}
\vskip 10pt 

Each side information hypergraph $\HH = (\V, \E_\HH)$ can be 
associated with the directed graph $\G_\HH = (\V,\E)$ in the following way. 
For each directed edge $(f(i), \X_i) \in \E_\HH$ there will be $|\X_i|$ directed edges $(f(i), v) \in \E$, for $v \in \X_i$. 
When $m = n$ and $f(i) = i$ for all $i \in [m]$, the graph $\G_\HH$ is, in fact, the \emph{side information graph}, 
defined in~\cite{Yossef}. 

Due to noise, the symbols received by $R_i$, $i \in [m]$, may be subject to errors.  
Assume that $S$ broadcasts a vector $\by \in \fq^N$. 
Let $\bep_i \in \fq^N$ be the error affecting the information received by $R_i$, $i \in [m]$. 
Then $R_i$ actually receives the vector
$\by_i = \by + \bep_i \in \fq^N$, instead of $\by$. 

\medskip
\begin{definition} 
Consider an instance of the ICSI problem described by $\HH = \HH(m,n,\X,f)$.
A \emph{$\delta$-error-correcting index code} ({\dd}) over $\fq$ for this instance is an encoding function
\begin{eqnarray*}
\fkE & : & \fq^n \rightarrow \fq^N \; , 
\end{eqnarray*}
such that for each receiver $R_i$, $i \in [m]$, there exists a decoding function
\[
\fkD_i \: : \: \fq^N \times \fq^{|\X_i|} \rightarrow \fq \; , \\
\]
satisfying
\[
\forall \bx, \bep_i \in \fq^n, \; \weight(\bep_i) \le \delta \; : \; \fkD_i(\fkE(\bx) + \bep_i, \bx_{\X_i}) = x_{f(i)}\; .
\]
\end{definition}

If $\delta=0$, we refer to such $\fkE$ as a \emph{non-error-correcting} index code, or just $\HH$-IC. 
The parameter $N$ is called the \emph{length} of the index code. 
In the scheme corresponding to the code $\fkE$, 
$S$ broadcasts a vector $\fkE(\bx)$ of length $N$ over $\fq$.   

\vskip 10pt 
\begin{definition} 
A \emph{linear index code} is an index code, for which the encoding function 
$\fkE$ is a linear transformation over $\fq$. 
Such a code can be described as 
\[
\forall \bx \in \fq^n \; : \; \fkE (\bx) = \bx \bL \; , 
\]
where $\bL$ is an $n \times N$ matrix over $\fq$. The matrix $\bL$ is called the \emph{matrix corresponding 
to the index code $\fkE$}, while $\fkE$ is referred to as the \emph{linear index code based on $\bL$}. 
\end{definition}
\vskip 10pt 

\begin{definition}
An \emph{optimal} linear {\dd} over $\fq$ is 
a linear {\dd} over $\fq$ of the smallest possible length $\NX$. 
\end{definition}
\vskip 10pt   

Hereafter, we assume that $\X = (\X_i)_{i \in [m]}$ is known to $S$.
We also assume that the code $\fkE$ is known to each receiver $R_i$, $i \in [m]$. 

\medskip
\begin{definition}
Suppose $\HH = \HH(m,n,\X,f)$ corresponds to an instance of the ICSI problem. 
Then the \emph{min-rank} of $\HH$ over $\ff_q$ is defined as 
\begin{multline*}
\kp_q(\HH) \define \min\{\rank_{\fq}(\{\bvi + \be_{f(i)}\}_{i \in [m]}) \; : \; \\ \bvi \in \fqn \; , \; \bvi \lhd \X_i\} \; .
\end{multline*} 
\end{definition} 
\medskip
Observe that $\kp_q(\HH)$ generalizes the $\mr$ over $\fq$ of the side information graph, which was defined in~\cite{Yossef}.
More specifically, when $m = n$ and $f(i) = i$ for all $i \in [m]$, $\G_\HH$ becomes the side information graph, 
and $\kp_q(\HH) = \mrq(\G_\HH)$. 
The $\mr$ was shown in~\cite{Yossef, Yossef-journal} to be the smallest number 
of transmissions in a linear index code. 

\vskip 10pt 
\begin{lemma}(\cite{Yossef,DauSkachekChee2010})
\label{lem:recovery}
Consider an instance of the ICSI problem described by $\HH = \HH(m,n,\X,f)$ . 
\begin{enumerate}
\item
The matrix $\bL$ corresponds to 
a linear {\mic} over $\fq$ if and only if for each $i \in [m]$ there exists $\bvi \in \fqn$ such that
$\bvi \lhd \X_i$ and $\bvi + \be_{f(i)} \in \clspn(\bL)$. 
\item The smallest possible length of a linear {\mic} over $\fq$ is $\kp_q(\HH)$. 
\end{enumerate}
\end{lemma}

\section{Basic Properties}
\label{subsec:ic_ecc_model}

We define the set of vectors
\begin{multline*}
\IX \define 
\left\{ \bz \in \fqn \; : \; \exists i \in [m] \text{ s.t. } \bz_{\X_i} = \bO, \; z_{f(i)} \neq 0 \right\}.
\end{multline*}
For all $i \in [m]$, we also define 
$\Y_i \define [n] \backslash \Big( \{f(i)\} \cup \X_i \Big)$. 
Then the collection of supports of all vectors in $\IX$ is given by
\begin{equation}
\label{Jdef} 
\JX \define \bigcup_{i \in [m]} \Big \{ \{f(i)\} \cup Y_i \; : \; Y_i \subseteq \Y_i \Big\}.
\end{equation} 

\begin{lemma}
\label{lem:decodability}
The matrix $\bL$ corresponds to a $(\delta, \HH)$-ECIC over $\fq$ if and only if 
\begin{equation} 
\label{2:E9}
\weight\left(\bz \bL\right) \geq 2\delta+1 \text{ for all } \bz \in \IX \; . 
\end{equation} 
Equivalently, $\bL$ corresponds to a $(\delta, \HH)$-ECIC over $\fq$ if and only if
\begin{equation}
\label{3:E9}
\weight\left(\sum_{i \in K}z_i\bL_i\right) \geq 2\delta + 1,
\end{equation}
for all $K \in \JX$ and for all choices of $z_i \in \fq^*$, $i \in K$. 
\end{lemma}
\begin{proof}
For each $\bx \in \fq^n$, we define
\[
B(\bx,\delta) = \{\by \in \fq^N \; : \; \by = \bx \bL + \bep, \; \bep \in \fq^N , \; \weight(\bep) \leq \delta \} \; ,
\] 
the set of all vectors resulting from at most $\delta$ errors in the transmitted vector 
associated with the information vector $\bx$. 
Then the receiver $R_i$ can recover $x_{f(i)}$ correctly
if and only if 
\[
B(\bx,\delta) \cap B(\bx',\delta) = \varnothing,
\]
for every pair $\bx,\bx' \in \fqn$ satisfying:
\[
\bx_{\X_i} = \bx'_{\X_i} \text{ and } x_{f(i)} \neq x'_{f(i)} \; . 
\]
(Observe that $R_i$ is interested only in the bit $x_{f(i)}$, not in the whole vector $\bx$.)

Therefore, $\bL$ corresponds to a $(\delta, \HH)$-ECIC if and only if the following condition is satisfied: 
for all $i \in [m]$ and for all $\bx,\bx' \in \fqn$ such that
$\bx_{\X_i} = \bx'_{\X_i}$ and $x_{f(i)} \neq x'_{f(i)}$, it holds 
\begin{multline}
\forall \bep, \bep' \in \fq^N, \; \weight(\bep) \le \delta, \; \weight(\bep') \le \delta \; : \\
 \bx \bL + \bep \neq \bx' \bL + \bep' \; . 
\label{eq:unique-decode}
\end{multline}
Denote $\bz = \bx' - \bx$. Then, the condition in~(\ref{eq:unique-decode}) can be reformulated as follows: 
for all $i \in [n]$ and for all $\bz \in \fqn$ such that $\bz_{\X_i} = \bO$ and $z_{f(i)} \neq 0$, it holds
\begin{multline}
\forall \bep, \bep' \in \fq^N, \; \weight(\bep) \le \delta, \; \weight(\bep') \le \delta \; : \;
\bz \bL \neq \bep -\bep' \; . 
\label{eq:unique-decode-z}
\end{multline}
The equivalent condition is that for all $\bz \in \IX$, 
\[
\weight(\bz \bL) \ge 2 \delta + 1 \; . 
\]
Inequality~(\ref{3:E9}) follows from this condition in a straight-forward manner. 
\end{proof}

\medskip
\begin{corollary}
\label{cor:decodability}
For all $i \in [m]$, let
\[
\bM_i \define \mbox{span} \left(  \{ \bL_j \; : \; j \in \Y_i \} \right) \; . 
\]  
Then, the matrix $\bL$ corresponds to a $(\delta, \HH)$-ECIC over $\fq$ if and only if  
\begin{equation} 
\label{cor:decod}
\forall i \in [m] \; : \; \dist(\bL_{f(i)}, \bM_i) \ge 2 \delta + 1 \; .  
\end{equation} 
\end{corollary}
\medskip

\medskip
\begin{example}
\label{example:1}
Let $q = 2$, $m = n=3$, and $f(i) = i$ for $i \in [3]$. 
Suppose $\X_1 = \{2,3\}$, $\X_2 = \{1,3\}$, and $\X_3 = \{1,2\}$. Let 
\[
\bL = \begin{pmatrix} 1 & 1 & 1 & 0 \\ 1 & 1 & 0 & 1\\ 1 & 0 & 1 & 1\end{pmatrix}. 
\]
Note that $\bL$ generates a $[4,3,1]_2$ code, which has minimum distance one. However, 
the index code based on $\bL$ can still correct one error. Indeed, let $\HH = \HH(3,3,\X,f)$, we have
\[
\I(2,\HH) = \{100,010,001\}. 
\]
Since each row of $\bL$ has weight at least three, it follows that 
$\weight(\bz \bL) \geq 3$ for all $\bz \in \I(2,\HH)$. By Lemma~\ref {lem:decodability}, 
$\bL$ corresponds to a $(1,\HH)$-ECIC over $\ff_2$. 
\end{example} 
\vskip 10pt 

\begin{example}
Assume that $m = n$ and $f(i) = i$ for all $i \in [m]$. Furthermore, suppose that $\X_i = \varnothing$ 
for all $i \in [m]$ (i.e. there is no side information available to the receivers). Let $\HH = \HH(m,n,\X,f)$.
Then, $\IX = \fqn \backslash \{ \bO \}$.
Hence, by Lemma~\ref {lem:decodability}, the $n \times N$ matrix $\bL$ corresponding to a {\dd} over $\fq$ 
(for some integer $\delta \ge 0$) is a generating matrix of an $[N,n, \ge 2 \delta+1]_q$ linear code. 
Thus, the problem of designing an ECIC is reduced to the problem of constructing a
classical linear error-correcting code.
\end{example}
\medskip 

\section{The $\al$-Bound and the $\kp$-Bound}
\label{subsec:bounds}

Let $(m,n,\X,f)$ be an instance of the ICSI problem, and let $\HH$ be the corresponding side information hypergraph. 
Next, we introduce the following definitions for the hypergraph $\HH$. 

\vskip 10pt 
\begin{definition}
A subset $H$ of $[n]$ is called a \emph{generalized independent set}
in $\HH$ if every nonempty subset $K$ of $H$ belongs to $\JX$. 
\end{definition}
\vskip 10pt 

\begin{definition}
A generalized independent set of the largest size in $\HH$ is called a \emph{maximum generalized independent set}.
	The size of a maximum generalized independent set in $\HH$ is called the \emph{generalized independence number}, 
	and denoted by $\al(\HH)$. 
\end{definition}
\vskip 10pt 

When $m = n$ and $f(i) = i$ for all $i \in [n]$, the
generalized independence number of $\HH$ is equal to the maximum size of an acyclic induced subgraph of $\G_\HH$, which was introduced in~\cite{Yossef}. In particular, when $\G_\HH$ is symmetric, $\al(\HH)$ is the 
independence number of $\G_\HH$. We omit the proof.    

\vskip 10pt
\begin{theorem}[$\al$-bound]
\label{thm:lowerbound} 
The length of an optimal linear $(\delta, \HH)$-ECIC over $\fq$ satisfies 
\[
\NX \ge \Na \; .
\]
\end{theorem}
\begin{proof}
Consider an $n \times N$ matrix $\bL$, which corresponds to a $(\delta, \HH)$-ECIC.
Let $H = \{i_1,i_2,\ldots,i_\aG\}$ be a maximum generalized independent set in $\HH$. 
Then, every subset $K \seq H$ satisfies $K \in \JX$. Therefore, 
\[
\weight\left( \sum_{i \in K} z_i \bL_i \right) \geq 2\delta + 1
\]
for all $K \subseteq H$, $K \neq \varnothing$, and for all choices of $z_i \in \fq^*$, $i \in K$. 
Hence, the $\aG$ rows of $\bL$, namely $\bL_{i_1},\bL_{i_2},\ldots,\bL_{i_\aG}$, form
a generator matrix of an $[N,\aG,2\delta+1]_q$ code. Therefore,
\[
N \geq \Na \; .
\]
\end{proof}

The following proposition is based on the fact that concatenation of 
a $\delta$-error-correcting code with an optimal (non-error-correcting) $\HH$-IC yields a $(\delta, \HH)$-ECIC. 

\medskip
\begin{proposition}[$\kp$-bound]
\label{pro:upperbound}
The length of an optimal $(\delta, \HH)$-ECIC over $\fq$ satisfies
\[
\NX \leq \Nk \; .
\]
\end{proposition}
The proof of this proposition is omitted due to lack of space. 
\vskip 3pt

\begin{corollary}
\label{coro:sandwiched}
The length of an optimal linear $(\delta, \HH)$-ECIC over $\fq$ satisfies
\[
\Na \leq \NX \leq \Nk \; .
\] 
\end{corollary}
\vskip 3pt

\begin{example}
\label{ex:1}
Let $q = 2$, $m = n = 5$, $\delta = 2$, and $f(i) = i$ for all $i \in [m]$. Assume
\begin{eqnarray*} 
\X_1 = \{2,5\} \; , \quad 
\X_2 = \{1,3\} \; , \quad 
\X_3 = \{2,4\} \; , \\
\X_4 = \{3,5\} \; , \quad
\X_5 = \{1,4\} \; . 
\end{eqnarray*}
Let $\HH = \HH(5,5,\X,f)$. 
The side information graph $\G_\HH$ of this instance is a pentagon. It is easy to verify that $\aG = \al(\G) = 2$. 
It follows from Theorem~9 in~\cite{Yossef-journal} that $\kp_2(\HH) = \mrt(\G_\HH) =~3$. 
Thus, from~\cite{Grassl} we have 
\[
N_2[2,5] = 8 \quad \mbox{ and } \quad N_2[3,5] = 10 \; . 
\]
Due to Corollary~\ref{coro:sandwiched}, we have
\[
8 \leq \N_2(\HH,2) \leq 10. 
\]
Using a computer search, we obtain that $\N_2(\HH,2) = 9$, and the corresponding optimal scheme is based on
\[
\bL = \begin{pmatrix} 
1  &1 & 1  &1 & 1&  0&  0 & 0  &0  \\
0  &1 & 0  &1  &1  &0  &1 & 1 & 0  \\
1  &1  &0  &0  &0  &1  &1  &1  &0  \\
0  &1  &1  &0  &0  &1  &0  &1  &1  \\
1  &0  &1  &0  &1  &0  &0  &1  &1  \\
\end{pmatrix} \; . 
\]
It is technical to verify that by Lemma~\ref {lem:decodability}, $\bL$
corresponds to $(2, \HH)$-ECIC. The length of this ECIC 
lies strictly between the $\al$-bound and the $\kp$-bound. 
\end{example}
\vskip 7pt 

\begin{remark}
Example~\ref{ex:1} illustrates that over small alphabets, the concatenation of an 
optimal linear (non-error-correcting) index code and an optimal linear error-correcting code may fail to 
produce an optimal linear ECIC. 
\end{remark}

\section{The Singleton Bound}
\label{sec:singleton}

\begin{theorem}[Singleton bound]
\label{thrm:singleton}
The length of an optimal  linear {\dd} over $\fq$
satisfies
\[
\NX \geq \kG + 2 \delta \; .
\]
\end{theorem}
\begin{proof}
Let $\bL$ be the $n \times \NX$ matrix corresponding to some optimal {\dd}. 
Let $\bL'$ be the matrix obtained by deleting any $2 \delta$ 
columns from $\bL$. 

By Lemma~\ref{lem:decodability}, $\bL$ satisfies
for all $\bz \in \IX$, 
\[
\weight(\bz \bL) \ge 2 \delta + 1 \; . 
\]
We deduce that the rows of $\bL'$ also satisfy that for all $\bz \in \IX$, 
\[
\weight(\bz \bL') \ge 1 \; . 
\]
By Lemma~\ref{lem:decodability}, $\bL'$ corresponds to a linear {\mic}.  
Therefore, by Lemma~\ref {lem:recovery}, part 2, $\bL'$ has at least $\kG$ columns. We deduce that
\[
\NX - 2 \delta \geq \kG \; ,
\]
which concludes the proof. 
\end{proof}
\vskip 10pt

The corollary below shows that for sufficiently large alphabets, a concatenation of a classical MDS error-correcting code
with an optimal (non-error-correcting) index code yields an optimal linear ECIC. 
\vskip 10pt
\begin{corollary} [MDS error-correcting index code]
For $q \geq \kG + 2 \delta - 1$, 
\begin{equation}
\NX = \kG + 2 \delta \; . 
\label{eq:mds}
\end{equation}
\end{corollary}
\begin{proof} Follows from Theorem~\ref{thrm:singleton} and Proposition~\ref{pro:upperbound}.
\end{proof} 

\medskip
\begin{remark}
There exist hypergraph $\HH$, such that $\G_\HH$ is the (symmetric) odd cycle of length $n$, 
for which the $\al$-bound is at least as good as the Singleton bound. 
\end{remark}

\section{Random codes}
\label{sec:random}

\begin{theorem}
\label{thrm:random}
Let $\HH = \HH(m,n,\X,f)$ describe an instance of the ICSI problem. 
Then there exists a {\dd} over $\fq$ of length $N$ if
\begin{equation}
\sum_{i \in [m]} q^{n - |\X_i| - 1} < \frac{q^N}{V_q(N, 2\delta)} \; , 
\label{eq:random-bound}
\end{equation}
where    
\[
V_q(N, 2\delta) = \sum_{\ell = 0}^{2 \delta} {N \choose \ell} (q-1)^\ell
\]
is the volume of the $q$-ary sphere in $\fq^N$. 
\end{theorem}
\emph{Idea of proof:}
We construct a random $n \times N$ matrix $\bL$ over $\fq$, row by row. 
Each row is selected independently of other rows, uniformly over $\fq^N$.
The result is obtained by bounding from above the probability of the event
\begin{eqnarray*}
\bigcup_{i \in [m]} E_i \; , \; \mbox{ where $E_i$} \; \define \; \left\{ \dist( \bL_{f(i)}, \bM_i ) < 2 \delta + 1 \right\} \; ,
\end{eqnarray*}
and by making this probability less than $1$.

\begin{remark}
The bound in Theorem~\ref{thrm:random} implies a bound on $\kp_q(\HH)$, which is tight for some $\HH$. 
Indeed, fix $\delta = 0$. Take $m = n = 2\ell + 1$ ($\ell \geq 2$), and $f(i) = i$ for all $i \in [n]$. 
Let $\X_1 = [n] \backslash \{1,2,n\}$ and $\X_n = [n] \backslash \{1,n-1,n\}$. 
For $2 \leq i \leq n - 1$, let $\X_i = [n] \backslash \{i-1,i,i+1\}$. 
Take $\HH = \HH(n,n,\X,f)$.  
Then $\G_\HH$ is the complement of the (symmetric directed) odd cycle of length $n$. 
We have $|\X_i| = 2\ell  - 2$ for all $i \in [n]$. 
Then~(\ref{eq:random-bound}) becomes
\[
N > 2 + \log_q(2\ell + 1) \; . 
\]
If $q > 2 \ell + 1$ then we obtain $N \ge 3$. Observe that in this case $\kp_q(\HH) = \mrq(\G_\HH)=3$ (see~\cite[Claim A.1]{Alon}), and thus the bound is tight.
\end{remark}

\section{Syndrome decoding}
\label{subsec:decoding}

Consider the $(\delta, \HH)$-ECIC based on a matrix $\bL$. 
Suppose that the receiver $R_i$, $i \in [m]$, receives the vector  
\begin{equation}
\by_i = \bx \bL + \bep_i \; ,
\label{eq:error-i}
\end{equation}
where $\bx \bL$ is the codeword transmitted by $S$, and $\bep_i$
is the error pattern affecting this codeword. 

In the classical coding theory, the transmitted vector $\bc$, the received vector $\by$, and the error pattern $\be$ are
related by $\by = \bc + \be$. For index coding, however, this is no longer the case. 
The following theorem shows that, in order to recover 
the message $x_{f(i)}$ from $\by_i$ using~(\ref{eq:error-i}), it is sufficient to find just one 
vector from a set of possible error patterns. This set is defined as follows: 
\[
\LL_i(\bep_i) = \left\{ \bep_i + \bz \; : \; \bz \in \spn(\{\bL_j\}_{j \in \Y_i}) \right\} \; . 
\] 
We henceforth refer to the set $\LL_i(\bep_i)$ as the \emph{set of relevant error patterns}. 

\vskip 10pt
\begin{lemma}
\label{thm:relevant_e}
Assume that the receiver $R_i$ receives $\by_i$.
\begin{enumerate}
\item
If $R_i$ knows the message $x_{f(i)}$ then it is able to 
determine the set $\LL_i(\bep_i)$. 
\item
If $R_i$ knows some vector $\beph \in \LL_i(\bep_i)$ then it is able to 
determine $x_{f(i)}$. 
\end{enumerate}
\end{lemma}
\vskip 10pt 

We now describe a syndrome decoding algorithm for linear error-correcting index codes.
We have
\[
\by_i - \bx_{\X_i}\bL_{\X_i} - \bep_i \in \spn\big(\{\bL_{f(i)}\} \cup \{\bL_j\}_{j \in \Y_i}\big) \; .
\]
Let $\cC_i = \spn(\{\bL_{f(i)}\} \cup \{\bL_j\}_{j \in \Y_i})$, and let $\bHi$ be a parity check matrix 
of $\cC_i$. We obtain that  
\begin{equation*} 
\bHi \bep_i^T = \bHi(\by_i - \bx_{\X_i}\bL_{\X_i})^T \; .
\label{eq:beta-i}
\end{equation*} 
Let $\bbti$ be a column vector defined by
\begin{equation*} 
\label{2:E12}
\bbti = \bHi(\by_i - \bx_{\X_i}\bL_{\X_i})^T \; .
\end{equation*} 
Observe that each $R_i$ is capable of determining $\bbti$. 
This leads us to the formulation of the decoding procedure for $R_i$ in Figure~\ref{fig:decoder}. 
{
\begin{figure}[htb]
\hrule
\vspace{2ex}
\begin{itemize}
  \item {\it Input:} $\by_i$, $\bx_{\X_i}$, $\bL$. 
  \vspace{1ex}
	\item {\it Step 1}: Compute the syndrome
	\[
	\bbti = \bHi(\by_i - \bx_{\X_i}\bL_{\X_i})^T \; . 
	\]  
  \item {\it Step 2}: Find the lowest Hamming weight solution $\hat{\bep}$ of the system
\begin{equation*}
\label{2:E13}
\bHi\hat{\bep}^T = \bbti \; .
\end{equation*} 
   \item {\it Step 3}: Given that $\bxh_{\X_i} = \bx_{\X_i}$, solve the system for $\hat{x}_{f(i)}$: 
	    \begin{equation*} 
			\label{2:E15}
			\by_i = \hat{\bx} \bL + \hat{\bep}.
			\end{equation*} 
   \item {\it Output:} $\hat{x}_{f(i)}$.
\end{itemize}
\vspace{2ex}
\hrule
\vspace{1ex}
\caption{Syndrome decoding procedure.}
\label{fig:decoder}
\end{figure}
}
\begin{theorem}
Let $\by_i = \bx \bL + \bep_i$ be the vector received by $R_i$, and let $\weight(\bep_i) \le \delta$. 
Assume that the procedure in Figure~\ref{fig:decoder} is applied to $(\by_i, \bx_{\X_i}, \bL)$. 
Then, its output satisfies $\hat{x}_{f(i)} = x_{f(i)}$. 
\end{theorem}

\begin{remark}
It is not impossible that $\hat{\bep} \neq \bep_i$. However, if $\weight(\bep_i) \leq \delta$, 
it can be shown that $\hat{\bep} \in \LL_i(\bep_i)$. Hence, by Lemma~\ref{thm:relevant_e}, 
we have $\hat{x}_{f(i)} = x_{f(i)}$. 
\end{remark}

\section{Acknowledgements}

The authors would like to thank the authors of~\cite{Yossef-journal} for providing a preprint of their paper. 
This work is supported by the National Research Foundation of Singapore (Research Grant
NRF-CRP2-2007-03).
 
\bibliographystyle{IEEEtran}
\bibliography{IndexCoding_ErrorCorrection}

\end{document}